\documentclass[letterpaper, 11pt]{article} 
\usepackage{bigstrut}
\usepackage{xparse}
\usepackage{float}
\usepackage{ifthen}
\usepackage{xspace}
\usepackage{setspace} 
\usepackage[margin=1in]{geometry}

\usepackage[ruled, linesnumbered]{algorithm2e}
\usepackage{adjustbox}
\usepackage{subfig}
\usepackage{amsmath}
\usepackage{amsfonts}
\usepackage{mathtools}
\usepackage{tikz}
\usetikzlibrary{arrows,calc,math,shapes.geometric,positioning}
\usepackage{forest}
\usepackage{amssymb}
\usepackage{multirow}

\usepackage{amsthm} 
\usepackage[bookmarksnumbered=true]{hyperref}
	\hypersetup{	
	colorlinks = true,
	linkcolor = blue,
	anchorcolor = blue,
	citecolor = blue,
	filecolor = blue,
	urlcolor = blue,
}
\usepackage[capitalise]{cleveref}  
\usepackage{chngcntr}

\newtheorem{theorem}{\bf{Theorem}}[section]

\newtheorem{cor}[theorem]{Corollary}
\newtheorem{lem}[theorem]{Lemma}

\newtheorem{example}[theorem]{Example}
\newtheorem{defn}[theorem]{Definition}

\theoremstyle{plain}

\crefname{equation}{}{}
\crefname{section}{Section}{Sections}
\crefname{subsection}{Subsection}{Subsections}
\crefname{figure}{Figure}{Figures}
\crefname{defn}{Definition}{Definitions}
\crefname{theorem}{Theorem}{Theorems}
\crefname{cor}{Corollary}{Corollaries}
\crefname{remark}{Remark}{Remarks}
\crefname{example}{Example}{Examples}

\newcommand{\MADst}{cMAD tree\xspace}

\title{\LARGE\bfseries 
High Dimensional Robust Consensus over Networks\\ with Limited Capacity}
\author{Yasin Yaz{\i}c{\i}o\u{g}lu\thanks{Yasin~Yaz{\i}c{\i}o\u{g}lu is with the Department of Electrical and Computer Engineering at the University of Minnesota, Minneapolis, MN, USA (e-mail: \url{ayasin@umn.edu}).}  \and Alberto Speranzon\thanks{Alberto Speranzon is with Honeywell Aerospace, Advanced Technology  (e-mail: \url{alberto.speranzon@gmail.com}).}
}

\begin{document}

\maketitle

\begin{abstract}
We investigate robust linear consensus over networks under capacity-constrained communication. The capacity of each edge is encoded as an upper bound on the number of state variables that can be communicated instantaneously. When the edge capacities are small compared to the dimensionality of the state vectors, it is not possible to instantaneously communicate full state information over every edge. We investigate how robust consensus (small steady state variance of the states) can be achieved within a linear time-invariant setting by optimally assigning edges to state-dimensions. \color{black}We show that if a finite steady state variance of the states can be achieved, then both the minimum cut capacity and the total capacity of the network should be sufficiently large. \color{black} Optimal and approximate solutions are provided for some special classes of graphs. We also consider the related problem of optimally allocating additional capacity on a feasible initial solution. We show that this problem corresponds to the maximization of a submodular function subject to a matroid constraint, which can be approximated via a greedy algorithm. 

\end{abstract}


\section{Introduction}
\label{sec:introduction}
Consensus networks are used to model the diffusive couplings in numerous natural and engineered systems, which typically operate in the face of various disturbances such as noise, communication delays, component failures, misbehaving nodes, or malicious attacks (e.g., \cite{Young10, Bamieh12,Shi13,Yasin15TNSE}). For noisy consensus networks, $\mathcal{H}_2$-type robustness measures can be defined based on the expected steady state variance of states (e.g., \cite{Young10, Bamieh12}).  Some graph-topological bounds on this robustness measure were presented in~\cite{Siami16,jadbabaie2013combinatorial}. A structural robustness measure, which is the smallest possible value of steady state variance under edge weights from the unit interval, was presented along with some tight bounds in~\cite{Yasin19CDC}. 

In this paper, we investigate the optimization of structural robustness to noise~\cite{Yasin19CDC} in linear consensus networks under capacity-constrained communications. We model the capacity of each edge as an upper bound on the number of state vector components that can be instantaneously communicated. When the edges have capacities smaller than the dimensionality of the state vectors, one way to achieve consensus is to design a switching system where the dimensions of state communicated over the edges change over time (e.g.,~\cite{farina2019randomized}). In this paper, we investigate the optimal time-invariant assignment of edges to the subsets of state-dimensions, which is simpler than switching policies to implement in a distributed manner with robustness guarantees. Our main contributions are as follows: \color{black}We show that for any connected network with $n$ nodes, each of which has a $k$-dimensional state vector, if a finite steady state variance of the states can be achieved with a time-invariant approach, then the minimum cut capacity of the network should be at least $k$ and the total edge capacity should be at least $k(n-1)$.  \color{black} We provide optimal solutions for complete graphs, where all the edge capacities are equal and have the minimum value required for a finite cost. We then present an approximate solution, which we show to be near-optimal via numerical simulations, for circulant graphs satisfying specific conditions on the edge capacities.  We also consider the problem of optimally improving a feasible initial solution by allocating spare capacity. We show that the optimal capacity allocation corresponds to the maximization of a monotone submodular function subject to a matroid constraint, which can be solved approximately by a greedy algorithm.   


\section{ Preliminaries}
\label{sec:prelim}

\subsection{Notation}
We denote by $\mathbb{R}_{+}$, $\mathbb{Z}_+$ the sets of non-negative reals, and integers, respectively. For a finite set $A$ with cardinality $|A|$, $\mathbb{R}^{|A|}$ represents the space of $|A|$-dimensional vectors. For pairs of vectors $x,y \in \mathbb{R}^{|A|}$, $x \leq y$ (or $x<y$) denotes the element-wise inequalities. 
The zero, $\mathbf{0}$, and all-ones, $\mathbf{1}$, vectors will have sizes that are clear from the context.

\subsection{Graph Theory Basics}
A graph $\mathcal{G}=(V,E)$ consists of a node set $V=\{v_1, v_2,\dots,v_n\}$ and an edge set $E \subseteq V \times V$. For an undirected graph, each edge is represented as an unordered pair of nodes. A \emph{path} between a pair of nodes $v_i,v_j \in V$ is a sequence of distinct nodes, $v_i, \dots, v_j$, such that each pair of consecutive nodes are linked by an edge. A graph is \emph{connected} if there exists a path between every pair of nodes.


Any connected graph with $n$ nodes has at least $n-1$ edges,  and any connected graph with exactly $n-1$ edges is called a \emph{tree}. A \emph{spanning tree} of a graph $\mathcal{G}$ is a subgraph that is a tree and includes all the vertices of $\mathcal{G}$.  A \emph{minimum average distance tree} (MAD tree) is a spanning tree of $\mathcal{G}$ that minimizes the average distance between all pairs of vertices of the graph~\cite{dankelmann2000average}. Given two trees, one rooted at $r$, $\mathcal{T}_r=(V,E)$, and one at $r'$, $\mathcal{T}_{r'}'=(V',E')$, we say they are \emph{isomorphic} if there exists a bijection $f: V \to V'$ such that $r'=f(r)$ and $(u,v) \in E$ if and only if $(f(u),f(v)) \in E'$.

A \emph{cut} $(U,W)$ is a partition of the nodes of a graph into two disjoint subsets $U,W \subset V$.  Any cut $(U,W)$ defines a unique \emph{cut-set}, which is the set of edges that are incident to one node in each subset of the partition, i.e., $\mathcal{E}(U,W)= \{(u,w) \in E \mid u\in U, \; w \in W  \}.$
For weighted graphs, we use $w \in \mathbb{R}_+^{|E|}$ to denote the vector of edge weights and $w_{ij}\in \mathbb{R}_+$ to denote the weight of the edge $(v_i,v_j)\in E$. Accordingly,  the (weighted) \emph{graph Laplacian}, $L_w$, satisfies $L_w \mathbf{1}=\mathbf{0}$ and has the entries $[L_w]_{ij}=-w_{ij}$ for any two adjacent nodes, $v_j\in \mathcal{N}_i$. We denote the unweighted Laplacian ($w=\mathbf{1}$) by $L$.

\subsection{Circulant Graphs}

A circulant graph encodes the abstract structure of a \emph{cyclic group} that is generated by a single element~\cite{hungerford12}. Given a set $G$, with a single associative binary operation, $+$, and a generator $g\in G$, every other element is obtained by applying the group operation to $g$ or its inverse. We consider only finite cyclic groups (finite graphs) of order~$n$, which are isomorphic to $\mathbb{Z}_n$ (the additive group of integers modulo~$n$). 

\begin{defn}[Circulant Graph]
	Let $\mathbb{Z}_n$ be a finite cyclic group of order $n$ and let $S\subseteq \mathbb{Z}_n\setminus\{0\}$ be inverse closed (if $s \in S$ then $-s \in S$). The undirected circulant graph has the vertex set $\{0,\dots,n-1\}$ and the edge set $\{(g,h) |(h-g)\in S\}$.

	Since the group $\mathbb{Z}_n$ and the set $S$ fully defines the graph, we will indicate a circulant graph as $\mathcal{G}=(\mathbb{Z}_n,S)$.
\end{defn}

For any cyclic group, a subset $S \subseteq \mathbb{Z}_n$ is called a \emph{generating set} if, for any $g\in \mathbb{Z}_n$, we have 
$
	g = \gamma_1 s_1 + \gamma_2 s_2 + \dots + \gamma_{|S|} s_{|S|} \pmod n
$,
where $\gamma_i\in \mathbb{Z}$ and $s_j \in S$\footnote{Note that we can sum node indexes and edges since $S\subseteq G$ and the $+ \pmod n$ operation is well defined.}. A circulant graph $(\mathbb{Z}_n,S)$ is (strongly) connected if and only if $S$ is a \emph{generating set} of $\mathbb{Z}_n$ \cite{godsil-royle13}. Since $S$ is inverse closed, $(\mathbb{Z}_n,S)$ is undirected. Paths from node $i$ to node $j$ on a circulant graph are defined by a vector $(\gamma_1,\gamma_2,\dots,\gamma_{|S|})\in \mathbb{Z}_+^d$ such that
$	j - i = \sum_{\ell=1}^{|S|} \gamma_\ell s_\ell \pmod n.$
In such a vector, each $\gamma_\ell$ defines how many edges of type $\ell$ belong to a path from~$i$ to~$j$. As the group is Abelian, the tuple $(\gamma_1,\dots,\gamma_{|S|})$ defines multiple paths each containing $\gamma_\ell$ edges of type $\ell$. 
Given $S$, we define $|S|/2$ \emph{classes} of edges. Each class corresponds to one element in $S$ and its inverse. Thus, we say that an edge is of class~$\ell$ if its  type is $\ell$ or $-\ell$ ($s_{\ell}$ or $s_{-\ell}$). 

\begin{example}
	As an example of a circulant graph, let us consider $\mathcal{G}=(\mathbb{Z}_9, \{\pm1, \pm 2, \pm3\})$.
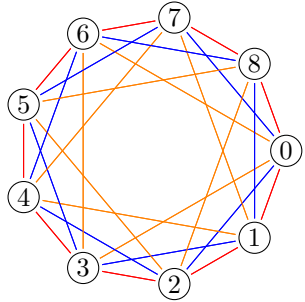
\begin{center}
	\begin{tikzpicture}[scale=1.2]
		\foreach \i in {0,...,8} \node[circle, draw, minimum size=2.5ex, inner sep = 0pt, outer sep = 0.05cm] (\i) at (-360/9 * \i:1.5cm) {\small $\i$};
		\foreach \i in {0,...,8}{
			\foreach \j in {+1,-1,+2,-2,+3,-3}{
				\pgfmathtruncatemacro{\nexti}{Mod(\i + \j,9)}
				\ifthenelse{\j=1 \OR \j=-1}{
					\draw[thin, red] (\i) -- (\nexti);}
				{
				 \ifthenelse{\j=2 \OR \j=-2}{
				 	\draw[thin, blue] (\i) -- (\nexti);}
			 	  {
				 	\ifthenelse{\j=3 \OR\j=-3}{
				 	\draw[thin, orange] (\i) -- (\nexti);}
				 	{\draw[thin, black] (\i) -- (\nexti);}} 	
		 		}
			}
		}
	\end{tikzpicture}
\end{center}
	Note that, since $\langle 1\rangle$ generates the group $\mathbb{Z}_9$, then the circulant graph is connected. Also, note that since $\langle 3 \rangle$ does not generate the group, the subgraphs comprising vertices $\{0,3,6\}$ and $\{1,4,7\}$ are not connected.
	We indicated with in red edges of class~$1$, in blue class~$2$ and in orange class~$3$.
\end{example}


\begin{defn}[Class-constrained MAD tree]\label{def:CMAD}
    Given a circulant graph $\mathcal{G}=(\mathbb{Z}_n,S)$ and $h\in\mathbb{Z}^{|S|/2}_+$, a \emph{class-constrained MAD tree} (\MADst) is a MAD tree of $\mathcal{G}$ that contains exactly $h_\ell$ edges of class~$\ell$. Note that such a \MADst may not exist for arbitrary $h\in\mathbb{Z}^{|S|/2}_+$.
\end{defn}

\subsection{Consensus Networks}
We present the preliminaries for a network, $\mathcal{G}=(V,E)$, where each agent $v_i\in V$ has a scalar state $x_i(t) \in \mathbb{R}$. Accordingly, the noisy consensus dynamics are expressed as \begin{equation}\label{eq:consensus2}
\dot{x}(t)=-L_w x(t)+ \xi(t),
\end{equation}
where $x_i(t) \in \mathbb{R}$ denotes the state of $v_i\in V$, $L_w$ denotes the weighted Laplacian matrix, and $\xi(t) \in \mathbb{R}^n$ is i.i.d. white Gaussian noise with zero mean and unit covariance. On any connected graph, the dynamics in \cref{eq:consensus2} result in a finite steady state variance of $x(t)$  \cite{Young10, Bamieh12}. Accordingly, the robustness of the network can be quantified through the expected population variance in steady state, i.e.,$
\mathcal{H}(\mathcal{G},w) \coloneqq \lim_{t \to \infty} \frac{1}{n}\sum\limits_{i =1}^{n} \mathrm{E}[{(x_i(t)-\tilde{x}(t))^2}]$, where $\tilde{x}(t) \in \mathbb{R}$ denotes the average of $x_1(t), \dots, x_n(t)$. It is possible to express $\mathcal{H}(\mathcal{G},w)$ in terms of the eigenvalues of $L_w$ (e.g., see \cite{Young10,Bamieh12}). In \cite{Yasin19CDC}, it was shown that $\mathcal{H}(\mathcal{G},w)$ monotonically decreases in all edge weights and the  smallest value under weights from the unit interval $(0,1]$ was defined as a structural measure of network robustness, i.e.,
\begin{equation}
\label{eq:hseig}
\mathcal{H}^*(\mathcal{G}) \coloneqq \min_{\mathbf{0} < w \leq \mathbf{1}} \mathcal{H}(\mathcal{G},w) = \frac{1}{2n} \sum_{i=2}^n \frac{1}{\lambda_i (L)},
\end{equation}
where ${0 < \lambda_2(L) \leq \dots \leq \lambda_n(L) }$ denote the eigenvalues of $L$. 


\section{Problem Formulation}
\label{sec:problem}
We consider a consensus network, where each agent $v_i$ has a $k$-dimensional state, i.e., $x_i(t) = [x_i^1(t) \dots x_i^k(t)]^\intercal \in \mathbb{R}^k$. The agents communicate with each other over an undirected network $\mathcal{G}=(V,E)$. Each edge $e \in E$ has some capacity $c_e \in \mathbb{Z}_+$, which is an upper bound on the dimensionality of vectors that can be communicated instantaneously over the edge. When $k \geq  c_e$, for some $e\in E$, the standard multi-dimensional consensus algorithm, where each dimension of state $x^1(t), \dots, x^k(t)$ evolves separately under \cref{eq:consensus2}, can not be implemented on $\mathcal{G}$.  Here, each $x^\ell(t) = [x_1^\ell(t) \dots x_n^\ell(t)]^\intercal  \in \mathbb{R}^n$ is the vector obtained by concatenating the $\ell^{th}$ dimension of each agent's state.   Accordingly, we investigate the problem of finding the optimal fixed assignment of edges to subsets of state-dimensions that minimizes the steady state population variance of states. Any such assignment of edges on $\mathcal{G}$ defines $k$ subgraphs, $\mathcal{G}_1, \dots, \mathcal{G}_k$, where each $\mathcal{G}_\ell =(V,E_\ell)$ denotes the network corresponding to the dynamics of the $\ell^{th}$ dimension of state:
$\dot{x}^\ell(t)=-L_w^\ell x(t)+ \xi^\ell(t)$. Here, $L_w^\ell$ is the weighted Laplacian of the graph with the edge set $E_\ell\subseteq E$, i.e., the set of edges over which the $\ell^{th}$ dimension of state is communicated.  As per the capacity limitations, each edge $e \in E$ can only appear in at most $c_e$ of these subgraphs. Hence,
\begin{equation}
   \label{eq:clim}
   \sum_{\ell=1}^k|\{e\} \cap E_\ell| \leq c_e, \; \forall e \in E.
\end{equation}



In this paper, we are investigating the optimal design of subgraphs $\mathcal{G}_1, \dots, \mathcal{G}_k$ to facilitate robust consensus, posed as the following optimization problem:
\begin{align}\label{eq:prob}
       \min_{\mathcal{G}_1=(V,E_1), \dots, \mathcal{G}_k=(V,E_k)} \quad &  \sum_{\ell=1}^k\mathcal{H}^*(\mathcal{G}_\ell) \\
       \mathrm{s.t.}
       \quad & \cref{eq:clim}. \nonumber
\end{align}


\section{Main Results}
\label{sec:main}
\subsection{Feasibility Condition}
\color{black}
 We start our analysis by showing that if \cref{eq:prob} has a solution with finite cost, then the minimum cut capacity of the network cannot be less than the dimensionality of the states. 

\begin{theorem}\footnote{\color{black} This is a correction to \cite[Thm 4.1]{yazicioglu2020high}, where  \cref{eq:minc2} was claimed to be also sufficient for the existence of a solution with finite cost. The sufficiency claim is retracted in this version with no impact on the other results of the paper.}
\label{thm:mincut-f}
Let $\mathcal{G}=(V,E)$ be a connected graph, and let each $e \in E$ have a capacity $c_e \in \mathbb{Z}_+$. If \cref{eq:prob} has a feasible solution $\mathcal{G}_1, \dots, \mathcal{G}_k$ with finite cost, i.e., 
\begin{equation}
\label{eq:minc1}
    \sum_{\ell=1}^k\mathcal{H}^*(\mathcal{G}_\ell) < \infty,
\end{equation}
then the minimum cut capacity of $\mathcal{G}$ satisfies
\begin{equation}
\label{eq:minc2}
    \min_{(U,W) \in \mathcal{C}(\mathcal{G})} \sum_{e \in \mathcal{E}(U,W)}c_e \geq k,
\end{equation}
where $\mathcal{C}(\mathcal{G})$ is the set of all cuts on $\mathcal{G}$, and $\mathcal{E}(U,W)$ is the cut-set associated with the cut $(U,W)$.
\end{theorem}
\begin{proof}

For any graph $\mathcal{G}_\ell$, let $L_\ell$ be the unweighted Laplacian. In light of \cref{eq:hseig}, $\mathcal{H}^*(\mathcal{G}_\ell)<\infty$ if and only if $0 < \lambda_2(L_\ell)$, which is true if and only if $\mathcal{G}_\ell$ is connected (e.g., see \cite{gross2003handbook}). Accordingly, \cref{eq:minc1} is true if and only if there exists $\mathcal{G}_1, \dots,\mathcal{G}_k$ that are all connected and satisfy \cref{eq:clim}. For any cut $(U,W)$, in order to have $U$ and $W$ connected on all $\mathcal{G}_1, \dots,\mathcal{G}_k$, each $E_1, \dots,E_k$ should contain at least one edge in the cut-set $\mathcal{E}(U,W)$. In light of \cref{eq:clim}, each $e \in \mathcal{E}(U,W)$ can be contained in at most $c_e$ of those edge sets. Accordingly, $U$ and $W$ must be disconnected on at least one of the graphs if  $\sum_{e \in \mathcal{E}(U,W)}c_e < k$. Hence, \cref{eq:minc2} is a necessary condition for the existence of a feasible solution to \cref{eq:prob} satisfying \cref{eq:minc1}.

\end{proof}


\color{black}
 Our next result builds on \cref{thm:mincut-f} and provides a necessary condition on the total edge capacity for \cref{eq:prob} to have a solution with finite cost. 
\begin{cor}
\label{eq:mincut-c1}
Let $\mathcal{G}=(V,E)$ be a connected graph, and let each $e \in E$ have a capacity $c_e \in \mathbb{Z}_+$. If \cref{eq:prob} has a feasible solution $\mathcal{G}_1, \dots,\mathcal{G}_k$ with finite cost, then the total edge capacity must be at least $k(n-1)$, i.e.,
\begin{equation}
\label{eq:mincc11}
    \sum_{e \in E }c_e \geq k(n-1).
\end{equation}
\end{cor}
\begin{proof}
As shown in the proof of Theorem~\ref{thm:mincut-f}, having a feasible solution $\mathcal{G}_1, \dots,\mathcal{G}_k$ with finite cost is equivalent to all of these subgraphs being connected. Any connected graph on $n$ nodes contains at least $n-1$ edges (e.g., see \cite{gross2003handbook}). Hence, if $\mathcal{G}_1, \dots, \mathcal{G}_k$ are all connected, then 
\begin{equation}
    \label{eq:mincc12}
    \sum_{\ell=1}^k |E_\ell| \geq k(n-1).
\end{equation}
Since each edge $e \in E$ can be contained in at most $c_e$ of the graphs $\mathcal{G}_1, \dots, \mathcal{G}_k$, we also have
\begin{equation}
    \label{eq:mincc13}
    \sum_{\ell=1}^k |E_\ell| \leq \sum_{e\in E} c_e.
\end{equation}
Using \cref{eq:mincc12} and \cref{eq:mincc13}, we obtain \cref{eq:mincc11}.
\end{proof}



\subsection{Special Classes of Graphs}\label{subsec:classes_graphs}  
Finding the optimal solution to  \cref{eq:prob} is a combinatorial problem that is intractable for arbitrary graphs. However, it is possible to identify the optimal or approximate solutions in some special cases.  In particular, we consider complete graphs and circulant graphs. We analyze complete graphs due to their significance as the networks with maximum robustness in the absence of capacity constraints. We also analyze circulant graphs, which contain complete graphs and other families such as cycle graphs and complete bipartite graphs, since they appear in many applications in engineering
and computer science (e.g.,~\cite{usevitch2017r,ekambaram2013circulant}).

\subsubsection{Complete Graph}
Our next result provides the optimal solution to \cref{eq:prob} for any complete graph where all the edge capacities are equal and they have the minimum value required for a finite cost, i.e., \cref{eq:mincc11} is satisfied with equality. 

\begin{theorem}
\label{thm:comp-t}
Let $\mathcal{G}$ be a complete graph with $n \geq 3$ nodes, let the state of each node be in $\mathbb{R}^k$ where $k=\alpha n$ for some  ${\alpha \in \mathbb{Z}_+}$, and let every edge have a capacity of $2\alpha$.
Then, \cref{eq:prob} has a unique optimal solution where $\mathcal{G}_1, \dots, \mathcal{G}_k$ are all star graphs and each node is the hub in exactly $\alpha$ of the stars.
\end{theorem}
\begin{proof}
For any $k=\alpha n$ and a complete graph where the capacity of edge is $2\alpha$, the total edge capacity satisfies 
\begin{equation}
\label{eq:comp-t1}
    \sum_{e \in E }c_e =\alpha n(n-1) = k(n-1).
\end{equation}
Accordingly, in light of~\cref{eq:mincut-c1}, such a network has the minimum possible total capacity for~\cref{eq:prob} to have a feasible solution with finite cost, i.e., $\mathcal{G}_1, \dots,\mathcal{G}_k$ that are all connected. Furthermore, since each $e \in E$ can be used in at most $c_e$ of those graphs, \cref{eq:comp-t1} implies that if a feasible solution with finite cost exists, then each of the corresponding $k$ subgraphs must have $n-1$ edges. Accordingly, $\mathcal{G}_1, \dots,\mathcal{G}_k$ must all be tree graphs. For tree graphs, the unique minimizer of \cref{eq:hseig} is the star graph (e.g., see \cite{Yasin19CDC}). Accordingly, we will complete the proof by showing that when the capacity of each edge on a complete graph is $2\alpha$, it is possible to build all of $\mathcal{G}_1, \dots, \mathcal{G}_k$ as star graphs and each node must be the hub in exactly $\alpha$ of those stars. To this end, first let us consider a set of $k$ star graphs where each node $v_i\in V$ is the hub in $\sigma_i$ of those graphs. Since there are $k=n\alpha$ graphs in total,
\begin{equation}
\label{eq:comp-t3}
    \sigma_1 +\sigma_2 + \dots + \sigma_n =n \alpha.
\end{equation}
Furthermore, the total capacity of the edges incident to any $v_i \in V$ should be at least equal to the sum of degrees of $v_i$ on the stars. As each $v_i \in V$ has a degree of $n-1$ on each star where $v_i$ is the hub and a degree of $1$ on all the other $k-\sigma_i=\alpha n-\sigma_i$ graphs, for all $i\in \{1, \dots, n\}$ we have
 \begin{equation}
\label{eq:comp-t4}
    \sum_{v_j:(v_i,v_j) \in E }\mkern-25mu c_{(v_i,v_j)} = 2\alpha(n-1)\geq \sigma_i(n-1)+\alpha n-\sigma_i.
\end{equation}
Using~\cref{eq:comp-t4}, we get 
$\alpha(n-2)\geq \sigma_i(n-2)$,
which implies that 
$\sigma_i \leq  \alpha$ for all $n\geq 3$. Using this inequality together with \cref{eq:comp-t3}, we obtain
 $\sigma_1= \dots = \sigma_n=\alpha$.
Accordingly, each $v_i$ must be the hub in exactly $\alpha$ star graphs. Note that the resulting graphs satisfy the capacity constraint of each edge since for every pair of nodes $v_i,v_j \in V$, $(v_i,v_j) \in E$ appears in $s_i+s_j =2\alpha$ of those stars. Consequently, building all of $\mathcal{G}_1, \dots, \mathcal{G}_k$ as star graphs, where each node is the hub of exactly $\alpha$ stars is the unique optimal solution to \cref{eq:prob}.
\end{proof}

\subsubsection{Circulant Graphs}
We present an approximate solution to~\cref{eq:prob} for circulant graphs satisfying some conditions on the edge capacities.  The proposed solution is based on finding a \MADst and using it to obtain $\mathcal{G}_1, \dots,\mathcal{G}_k$ that are all related through the rotation isomorphism, which we define next. 

\begin{defn}[Rotation Isomorphism]\label{def:rotation_iso} 
    Given a tree rooted at an arbitrary node~$i$, $\mathcal{T}_i$, of a circulant graph, we define a \emph{rotation by $\delta$} a new tree $\mathcal{T}_j$ rooted at $j=i+\delta \pmod n$ such that every $v\in \mathcal{T}_i$ is mapped to $w = v + \delta \pmod n$ in $\mathcal{T}_j$. Such a  rotation defines an isomorphism.
\end{defn}

\begin{lem}\label{lem:existance_MAD_circ_graph}
Let $\mathcal{G}=(\mathbb{Z}_n,S)$ be a circulant graph and let $h \in \mathbb{Z}^{|S|/2}_+$ satisfy $\sum_{\ell=1}^{|S|/2} h_\ell = n - 1$, with $h_{\bar{\ell}} \geq 1$ for at least one generator $\bar{\ell}$. Then, there exists a cMAD with exactly $h_\ell$ edges of class~$\ell$. 
\end{lem}
\begin{proof}
We show that $\mathcal{G}$ has a spanning tree that contains exactly $h_\ell$ edges of each class~$\ell$. The \MADst is then such a spanning tree that achieves the minimum average distance. Let us assume, without loss of generality, that the generator of the group is $s_1=1$. Indeed we can always map such a generator to a different one~\cite{hungerford12}. We know that by taking $n-1$ distinct integers we can generate all the elements of the group $\mathbb{Z}_n$. Thus we need $n-1$ edges of class~$1$ to generate the group. This corresponds to the line graph on $\mathcal{G}$. Now, let us assume that we remove one edge of class~$1$ and include a new edge of class~$\ell'$, for arbitrary~$\ell'$. Thus we have $n-2$ edges of class~$1$ and one edge of class~$\ell'$. By using $n-2$ edges of class~$1$ we cannot generate the full group. Let us say the group element~$q$ is not generated. However, we can always generate~$q$ by considering the extra edge of class~$\ell'$ and by adding the edge $(p,q)$ where $p = q + s_{\ell'} \pmod n$. Of course we can continue this construction by changing the number $h_\ell$  of edges of class~$\ell$ as long as $\sum_{\ell=1}^{|S|/2} h_\ell=n-1$. We need to retain at least an edge corresponding to the generator class, to ensure connectivity. As the group is Abelian the order in which edges are removed/added is irrelevant. Thus, given $h_\ell$ we can generate multiple trees by choosing which specific edge to remove/add. Among the trees with $h_\ell$ edges of class~$\ell$, the \MADst is the one with minimum average distance.
\end{proof}
%

Before we state the main theorem regarding the proposed solution for circulant graphs, we introduce a lemma that proves certain properties of the \emph{rotation} isomorphism, which is defined in \cref{def:rotation_iso} and used in line~7 of~\cref{alg:algoritm_circ_graph}.

\begin{lem}\label{lem:tech_lemma_cayley_proof}
Given a circulant graph $\mathcal{G} = (\mathbb{Z}_n,S)$  let $\mathcal{T}_i$ be a \MADst rooted at an arbitrary node~$i$ with the number of edges of class~$\ell$ being $h_\ell$. Consider the set of $n$ trees $\{\mathcal{T}_j\}_{0}^{n-1}$, where $\mathcal{T}_{j\neq i}$ is obtained using the \emph{rotation by j} defined in~\cref{def:rotation_iso}. Then an edge $(p,q)$ of $\mathcal{G}$ of class~$\ell$ appears in $h_\ell$ of the trees $\{\mathcal{T}_j\}_{0}^{n-1}$. 
\end{lem}

\begin{proof}
	Let us consider an edge $(p,q)$ of class~$\ell$. We need to find how many edges $(s,t)$ of class~$\ell$ in $\mathcal{T}_{0,\dots,n-1}$ can be mapped, via the isomorphism, to $(p,q)$. For this to happen we are looking for the values of $j<n$ for which we either have $(s+j-i,t+j-i) \pmod n = (p,q)$ or $(s+j-i,t+j-i) \pmod n =(q,p)$, as the graph is undirected. Since there is an isomorphism (bijection) between every two trees, there is a unique solution that satisfies one of the two previous equations, for every $j<n$.
	Because there are $h_\ell$ edges of class~$\ell$ in each tree, then we have that any edge of class~$\ell$ in $\mathcal{G}$ will appear in $h_\ell$ of the $n$ trees $\{\mathcal{T}_j\}_{0}^{n-1}$.
	%
\end{proof}

\begin{example}
As mentioned that the isomorphism we have considered performs a \emph{rotation} of a tree over the circulant graph. Let us consider the following example, with $\mathcal{G}=(\mathbb{Z}_9,\{\pm1,\pm4\})$, to clarify this point. The isomorphism we have introduced in~\cref{lem:tech_lemma_cayley_proof} rotates a tree by $j=1$ clockwise, in this example, as the picture below shows.

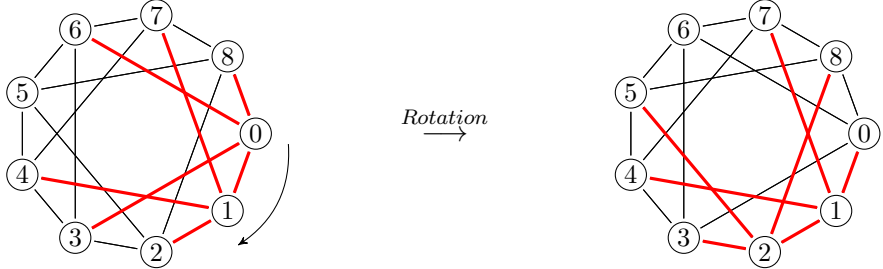
\begin{minipage}{.40\columnwidth}
			\centering\noindent
				\begin{tikzpicture}[scale=0.8]
					\foreach \i in {0,...,8} \node[circle, draw, minimum size=2.5ex, inner sep = 0pt, outer sep = 0.05cm] (\i) at (-360/9 * \i:2cm) {\small $\i$};
					\foreach \i in {0,...,8}{
						\foreach \j in {+1,-1,+3,-3}{
							\pgfmathtruncatemacro{\nexti}{Mod(\i + \j,9)}
							\draw[thin, black] (\i) -- (\nexti);
						}
					}
					\draw[very thick, red] (0) -- (1);
					\draw[very thick, red] (0) -- (8);
					\draw[very thick, red] (0) -- (3);
					\draw[very thick, red] (0) -- (6);
					\draw[very thick, red] (1) -- (2);
					\draw[very thick, red] (1) -- (4);
					\draw[very thick, red] (1) -- (7);
					\node (A) [right=0.03cm of 0] {};
					\node (B) [below=0.15cm of 1] {};
					\path (A) edge [->, >=stealth',auto, bend left] node {} (B);
				\end{tikzpicture}
		\end{minipage}$\stackrel{\scriptsize Rotation}{\longrightarrow}$
		\begin{minipage}{.4\columnwidth}
			\centering\noindent
			\begin{tikzpicture}[scale=0.8]
					\foreach \i in {0,...,8} \node[circle, draw, minimum size=2.5ex, inner sep = 0pt, outer sep = 0.05cm] (\i) at (-360/9 * \i:2cm) {\small $\i$};
					\foreach \i in {0,...,8}{
						\foreach \j in {+1,-1,+3,-3}{
							\pgfmathtruncatemacro{\nexti}{Mod(\i + \j,9)}
							\draw[thin, black] (\i) -- (\nexti);
						}
					}
					\draw[very thick, red] (1) -- (2);
					\draw[very thick, red] (1) -- (0);
					\draw[very thick, red] (1) -- (4);
					\draw[very thick, red] (1) -- (7);
					\draw[very thick, red] (2) -- (3);
					\draw[very thick, red] (2) -- (5);
					\draw[very thick, red] (2) -- (8);
				\end{tikzpicture}
		\end{minipage}
\end{example}

\begin{algorithm}[t!]
	\KwData{
		Cyclic group order $n$ and generating set $S$\;
		Capacity for class~$\ell$ edges set to $h_\ell\alpha$, with $h_\ell$ such that $\sum_{i=1}^{|S|/2} h_\ell = n - 1$ and $h_{\bar{\ell}}\geq 1$ for at least one generator $s_{\bar{\ell}}$.
	}
	\KwResult{Graphs $\mathcal{G}_1,\dots,\mathcal{G}_{\alpha n}$}
	\smallskip
	$\mathcal{G} \leftarrow \mathsf{CreateCircGraph}(\mathbb{Z}_n,S)$\;
	\emph{{$\triangleright$} Compute a \MADst with number of edges of class~$\ell$ being $h_\ell$ and pick an arbitrary node $i$ as a root}\;
	$\mathcal{T}_i \leftarrow \mathsf{cMAD}(\mathcal{G}, (h_1,\dots,h_\ell))$\;
	\For{$j\leftarrow 0$ \KwTo $n-1$}{
	    \If{$j\neq i$}{
			\emph{{$\triangleright$} Create trees from $\mathcal{T}_i$ via \emph{rotation} (\cref{def:rotation_iso})}\;
			$\mathcal{T}_{j} \leftarrow \mathsf{Rotate}(\mathcal{T}_i,\delta=j-i)$\;
		}
		$(\mathcal{G}_{j\alpha+1},\dots,\mathcal{G}_{(j+1)\alpha}) \leftarrow \mathsf{CreateCopies}(\mathcal{T}_j,\alpha)$\;
	}
	\caption{\small Optimal Spanning Trees of a Circulant Graph}\label{alg:algoritm_circ_graph}
\end{algorithm}

\begin{theorem}\label{thm:cayley_graphs_capacity}
Let $\mathcal{G}=(\mathbb{Z}_n,S)$ be a circulant graph. Let the state of each node be in $\mathbb{R}^k$ where $k=\alpha n$ for some  ${\alpha \in \mathbb{Z}_+}$. Let $S = \{\pm s_1, \pm s_2, \dots, \pm s_{|S|/2}\}$ and let the capacity of each edge of class $\ell$ be $h_\ell\alpha$ where $h_\ell\in \mathbb{Z}_+$ satisfy 
\begin{align}\label{eq:cayley_capacity_constraint}
	\sum_{\ell=1}^{|S|/2} h_\ell = n-1\,, 
\end{align}
and $h_{\bar{\ell}}\geq 1$ for at least one generator $s_{\bar{\ell}}$.
Then, \cref{alg:algoritm_circ_graph} generates a solution  $\mathcal{G}_1,\dots,\mathcal{G}_k$ to~\cref{eq:prob}.
\end{theorem}

\begin{proof}
	Given a circulant graph, where all edges in the same class have the same capacity, we have
		\begin{align}\label{eq:capacity_Cayley}
		\sum_{e\in E} c_e = n \sum_{\ell=1}^{|S|/2} c_\ell\,,
	\end{align}
	where $c_\ell$ denotes the capacity of edges of class $\ell$.
    Note that \cref{eq:capacity_Cayley} holds as we assumed that edges in the class~$\ell$---$s_\ell$ and $s_{-\ell}$---have the same capacity. Indeed, an edge in the class~$\ell$ appears in the graph whenever there are group elements $g,h \in \mathbb{Z}_n$ such that $g-h = s_\ell \pmod n$. There are exactly~$n$ distinct pairs of elements in $\mathbb{Z}_n$  whose difference is $s_\ell$.  
	
	As $c_\ell = h_\ell\alpha$ and $k=\alpha n$, by using  \cref{eq:cayley_capacity_constraint}  and \cref{eq:capacity_Cayley}, we get
	\begin{align}\label{eq:cayley_total_capacity}
		n \sum_{\ell=1}^{|S|/2} c_\ell &=\alpha n \sum_{\ell=1}^{|S|/2} h_\ell = n (n-1)\alpha = k(n-1)\,.
	\end{align}
		
	Thus, a circulant graph with such capacities achieves the minimum total capacity requirement for the optimization problem~\cref{eq:prob} to have a finite solution (see~\cref{eq:mincut-c1}). Furthermore,~\cref{eq:cayley_total_capacity} implies that the graphs $\mathcal{G}_1,\dots,\mathcal{G}_k$ must have $n-1$ edges, namely each needs to be a tree. In light of \cite[Thm 3.2]{Yasin19CDC}, for any tree $\mathcal{G}_\ell$, we have
	$\mathcal{H}^*(\mathcal{G}_\ell) = \tilde{\delta}(\mathcal{G}_\ell)(n-1)/4n,$
	where $\tilde{\delta}(\mathcal{G}_\ell)$ is the average distance of $\mathcal{G}_\ell$. Accordingly, we are left to show that we can construct $\mathcal{G}_1,\dots,\mathcal{G}_k$ that are all \MADst and satisfy the capacity constraint in \cref{eq:clim}.
	
	Recall from~\cref{lem:existance_MAD_circ_graph} a \MADst has $h_\ell$ edges of class~$\ell$, such that~\cref{eq:cayley_capacity_constraint} is satisfied. Furthermore, from~\cref{lem:tech_lemma_cayley_proof}, we know that if we assign to each node~$i\in\mathcal{G}$ a \MADst, then an edge $(p,q)\in \mathcal{G}$ of class~$\ell$ will belong to $h_\ell$ trees. Thus if we take $\alpha$ copies of $\mathcal{T}_i$ then an edge $(p,q)$ of class~$\ell$ will belong to $h_\ell\alpha$ trees. Thus by considering $\mathcal{G}_i,\dots,\mathcal{G}_{i+\alpha}$ to be $\alpha$ copies of the \MADst rooted at~$i$, the set $\mathcal{G}_1,\dots,\mathcal{G}_{\alpha n}$ achieves the required capacity. 
\end{proof}

\subsection{Optimality and Complexity of~\cref{alg:algoritm_circ_graph}}

For circulant graphs satisfying the premise of~\cref{thm:cayley_graphs_capacity}, any feasible solution to \cref{eq:prob} must be a set of trees $\mathcal{G}_1,\dots,\mathcal{G}_k$ (as per the proof of \cref{thm:cayley_graphs_capacity}). For each tree $\mathcal{G}_i$, $\mathcal{H}^*(\mathcal{G}_i)$ is uniquely defined by the average distance among nodes on $\mathcal{G}_i$ (e.g., see \cite{Yasin19CDC}). Hence, solving \cref{eq:prob} becomes equivalent to finding a feasible combination of spanning trees that minimize the sum of average distances among nodes.
Hence, by the definition of MAD trees, the cost in \cref{eq:prob} can never be lower than $k\mathcal{H}^*(\mathcal{G}^{\mathrm{MAD}})$, where $\mathcal{G}^{\mathrm{MAD}}$ is a MAD tree of $\mathcal{G}$.  Given the capacity constraints on edges of class~$\ell$, let $\mathcal{G}^{\mathrm{cMAD}}$ be the \MADst provided by~\cref{alg:algoritm_circ_graph}. Furthermore, let $\mathcal{G}_1^{\mathrm{Opt}}, \dots,\mathcal{G}_k^{\mathrm{Opt}}$ be the optimal solution to \cref{eq:prob}. Then, we have
\begin{align}\label{eq:order_solutions}
     k\mathcal{H}^*(\mathcal{G}^{\mathrm{MAD}}) \leq  \sum_{i=1}^k  \mathcal{H}^*(\mathcal{G}_i^{\mathrm{Opt}}) \leq  k\mathcal{H}^*(\mathcal{G}^{\mathrm{cMAD}}).
\end{align}
In general, the capacity constraints may not allow for a combination of $k$ MAD trees to be a feasible solution to \cref{eq:prob} and the left inequality in \cref{eq:order_solutions} can be strict. However, both inequalities in~\cref{eq:order_solutions} become equalities whenever the capacity constraints are such that the {\MADst}s are also MAD trees. In such cases, the solution generated by \cref{alg:algoritm_circ_graph} is optimal. One such example is the complete graphs satisfying the premise in \cref{thm:comp-t}, where \cref{alg:algoritm_circ_graph} would return the unique optimal solution described therein. Another example is the family of distance hereditary graphs, where every connected induced subgraph has the same pairwise distances between nodes as in the original graph. Note that the {\MADst}s being MAD trees is sufficient but not necessary to obtain optimal solutions via ~\cref{alg:algoritm_circ_graph}. In  ~\cref{sec:sim_results},  we also provide numerical results with various circulant graphs to show that  ~\cref{alg:algoritm_circ_graph} generates optimal or near-optimal solutions in most cases when the premise of \cref{thm:cayley_graphs_capacity} is satisfied. 

One of the main advantages of~\cref{alg:algoritm_circ_graph} is that it requires to find a single \MADst, which is then rotated and replicated to produce $\mathcal{G}_1,\dots,\mathcal{G}_k$. This approach significantly reduces the computational load compared to searching for an optimal combination of $k$ spanning trees that together satisfy the capacity constraints. However, finding a \MADst is still a computationally challenging task for arbitrary circulant graphs. Finding the unconstrained version (a MAD tree) is in general NP-hard~\cite{johnson1978complexity}. For special cases such as  distance hereditary circulant graphs, i.e., circulant graphs where every cycle of length five or more has at least two crossing diagonals (e.g.,  $\mathcal{G}=(\mathbb{Z}_6,\{\pm1, \pm3\})$), a MAD/cMAD tree can be computed in polynomial time~\cite{bandelt1986distance}. Accordingly,~\cref{alg:algoritm_circ_graph} provides the optimal solution in polynomial time for such special cases. For the general case, an approximate \MADst satisfying the capacity constraints can be obtained by adapting a polynomial time algorithm (e.g., \cite{wu2000polynomial}) to use in ~\cref{alg:algoritm_circ_graph}.

\subsection{Improving an Initial Solution via Capacity Allocation}
We analyze next the problem of optimally allocating additional/remaining capacity to an initial set of connected $\mathcal{G}_1^0, \dots,\mathcal{G}_k^0$ inside the feasible space of~\cref{eq:prob}. We are interested in this problem with two motivations. Firstly, given the intractability of the combinatorial problem in~\cref{eq:prob} for arbitrary graphs, a practical solution approach is to first find a good combination of spanning trees and then to allocate the remaining capacity. For example, for complete graphs that have more capacity than required in \cref{thm:comp-t}, a solution can be obtained by allocating the remaining capacity over the star graphs specified therein.  Secondly, if more capacity becomes available (e.g., new edges), it is more practical to improve an existing solution than redesigning the whole system. 

Given a network $\mathcal{G}=(V,E)$, let $\mathcal{E}=E \times \{1,2, \dots,k\}$ be the set of edge-dimension pairs, where each $(e,\ell) \in \mathcal{E}$ corresponds to edge $e$ in graph $\mathcal{G}_\ell$. Accordingly, $\mathcal{E}$ can be represented as the union of disjoint sets, i.e.,
\begin{equation}
\label{eq:part1}
   \mathcal{E}= \bigcup_{e\in E} \mathcal{E}_e\,,\quad \text{where} \quad \mathcal{E}_e = \{(\epsilon,\ell) \in \mathcal{E} \mid \epsilon=e\}\,.
\end{equation}
Let $\mathcal{G}_1^0, \dots,\mathcal{G}_k^0$ be connected subgraphs of $\mathcal{G}$ that satisfy \cref{eq:clim}, 
and for each $e\in E$, let $\tilde{c}_e$ denote the remaining capacity,
\begin{equation}
\label{eq:crem}
     \tilde{c}_e = c_e-  \sum_{\ell=1}^k|\{e\} \cap E_\ell^0|\,.
\end{equation}
 Furthermore, let $\mathcal{E}^0$ be the edge-dimension pairs encoded in the initial graphs $\mathcal{G}_1^0, \dots,\mathcal{G}_k^0$: $ \mathcal{E}^0 = \{(e,\ell) \in \mathcal{E} \mid e \in E_\ell^0\}$. Accordingly, in light of \cref{eq:hseig}, finding an optimal allocation of remaining capacity to solve \cref{eq:prob} can be formulated as 
\begin{align}\label{eq:prob2}
       \max_{\mathcal{E}'\subseteq \mathcal{E}\setminus \mathcal{E}^0} \quad &  -\sum_{\ell=1}^k \sum_{i=2}^n \frac{1}{\lambda_i (L'_\ell)} \\
       \mathrm{s.t.} \quad & |\mathcal{E}' \cap \mathcal{E}_e| \leq \tilde{c}_e, \; \forall e \in E,  \nonumber
\end{align}
where $L'_\ell$ is the Laplacian matrix associated with the graph $\mathcal{G}_\ell'=(V,E_\ell^0 \cup E_\ell')$ and $E_\ell'=  \{e \in E \mid (e,\ell)\in \mathcal{E}'\}$.
\begin{theorem}
    \label{thm:submod}The problem in \cref{eq:prob2} is the maximization of a monotone submodular function under a matroid constraint.
\end{theorem}
\begin{proof}
Given a network $\mathcal{G}=(V,E)$ and a set of connected subgraphs $\mathcal{G}_1^0, \dots,\mathcal{G}_k^0$,  it follows from \cite[Theorem 3]{summers2015topology} that $\phi_\ell: 2^{\mathcal{E} \setminus \mathcal{E}^0} \mapsto \mathbb{R}$, with
 $    \phi_\ell (\mathcal{E}')= - \mathcal{H}^*(\mathcal{G}'_\ell)$, is submodular. 
Furthermore, in \cite[Theorem 2.7]{Ellens11} it was shown that the effective resistance, which equals $n\sum_{i=2}^n 1/\lambda_i (L'_\ell)$, strictly decreases with additional edges in the corresponding graph. Hence, $\phi_\ell (\mathcal{E}')$ is a monotone submodular function. Thus, the objective in \cref{eq:prob2}, which equals $\sum_{\ell=1}^k\phi_\ell (\mathcal{E}')$, is also monotone submodular. As \cref{eq:part1}, the constraint of \cref{eq:prob2}, is a standard partition matroid, we complete the proof.
\end{proof}
In light of \cref{thm:submod}, a greedy algorithm can be used to obtain a 1/2-approximation to \cref{eq:prob2} (e.g., see  \cite{krause2014submodular}). Starting with $\mathcal{E}'=\emptyset$, such a greedy algorithm grows $\mathcal{E}'$ iteratively by adding an element from  $\mathcal{E}\setminus (\mathcal{E}^0 \cup \mathcal{E}')$ that maximally increases the objective while satisfying the constraint in \cref{eq:prob2}. This procedure continues until no more elements can be added to $\mathcal{E}'$, i.e., all the remaining capacity is allocated.

\section{Numerical Results}\label{sec:sim_results}

We compare $\mathcal{H}^*(\mathcal{G}^{\mathrm{cMAD}})$ to $\mathcal{H}^*(\mathcal{G}^{\mathrm{MAD}})$ to demonstrate an upper bound on the optimality gap. More specifically, we consider the ratio $$\Delta = \frac{\mathcal{H}^*(\mathcal{G}^{\mathrm{cMAD}})-\mathcal{H}^*(\mathcal{G}^{\mathrm{MAD}})} {\mathcal{H}^*(\mathcal{G}^{\mathrm{cMAD}})}\,.$$
\cref{tbl:comparisons} provides the results for various graphs.  For each circulant graph with $n$ nodes and the classes of edges, $\ell$, we have considered all the possible values of $h_\ell$ that satisfy~\cref{eq:cayley_capacity_constraint}. Among those values of $h_\ell$, we provide the results for the ones that yield the best and worst values of $\mathcal{H}^*(\mathcal{G}^{\mathrm{cMAD}})$. The small values of $\Delta$ in \cref{tbl:comparisons} show that Alg. 1 provides optimal or near-optimal solutions in all of these cases.

\section{Conclusion}
\label{sec:conclusion}
We have introduced a robust linear consensus problem over networks with capacity constraints bounding the number of state variables that can be communicated instantaneously over the edges. We investigated the optimal assignment of edges to subsets of state-dimensions for maximizing robustness. \color{black}We showed that for any connected network with $n$ nodes, each of which has a $k$-dimensional state vector, if a finite steady state variance of the states can be achieved via a time-invariant assignment of edges to state-dimensions, then the minimum cut capacity of the network should be at least $k$ and the total edge capacity should be at least $k(n-1)$.  \color{black} Under certain conditions on the edge capacities, we provided the optimal solution for complete graphs and an approximate solution for circulant graphs.  The latter is optimal in certain cases and we also provided numerical results showing its near-optimality in various examples. We also considered the problem of optimally allocating additional capacity to improve a feasible initial solution. We showed that this problem corresponds to the maximization of a submodular function subject to a matroid constraint, which can be solved approximately via a greedy algorithm.

\begin{table}
	\centering
\begin{tabular}{c|c|c|c|c|c}
		$n$ & $\ell$ & $h_\ell$ & $\mathcal{H}^*(\mathcal{G}^{\mathrm{MAD}})$  & $\mathcal{H}^*(\mathcal{G}^{\mathrm{cMAD}})$ & $\Delta$ \\
		\hline
		\multirow{4}{*}{10} & \multirow{2}{*}{3,5}	& 5,4 	 & \multirow{2}{*}{0.605} & 0.605   &  0\%\\
							&  & 8,1    & & 0.645   &  6.0\%  \\
		                    \cline{2-6}
							& \multirow{2}{*}{1,3,4} & 3,3,3  & \multirow{2}{*}{0.481} & 0.481   &  0\%   \\
                            & & 1,7,1  &  & 0.570   &  15.6\% \\
        \hline
        \multirow{4}{*}{15} 
        					& \multirow{2}{*}{1,3,4} & 4,6,4  & \multirow{2}{*}{0.622} & 0.622   & 0\%	\\
        					& & 7,3,4  &  & 0.653   & 4.7\%	\\
        				     \cline{2-6}
        					& \multirow{2}{*}{3,4,7} & 4,4,5  & \multirow{2}{*}{0.640} & 0.640   & 0\% 	\\
        					& & 3,1,10  &  & 0.707   & 9.5\% \\
        \hline
        \multirow{2}{*}{20} & \multirow{2}{*}{1,2,4} & 8,6,5  & \multirow{2}{*}{0.734} & 0.734  & 0\%\\
                            &  & 2,2,15  &  & 0.801   & 8.4\%	\\
      \hline
        \multirow{2}{*}{30} & \multirow{2}{*}{1,4,5} & 9,10,10 &  \multirow{2}{*}{0.908} & 0.908 & 0\%\\
                            & & 4,16,9 & & 0.934 & 2.7\%\\
       \hline
        \multirow{2}{*}{35} & \multirow{2}{*}{1,6,7,10} & 10,4,10,10 &  \multirow{2}{*}{0.770} & 0.770 & 0\%\\
                            &  & 6,8,15,5 & & 0.783 & 1.7\%
    \end{tabular}\caption{$n$ is the number of nodes, $\ell$ is the classes of edges and $h_\ell$ is the number of edges (each with a capacity of $h_\ell \alpha$) in each class. The observed values of the robustness measure in ~\cref{eq:hseig} for MAD and cMAD trees, and the maximum and minimum of their relative difference, $\Delta$, are shown. 
    }\label{tbl:comparisons}
\end{table}

\bibliography{MyReferences}

\end{document}